\documentclass[11pt]{article}
\usepackage[a4paper, margin=1in]{geometry}
\usepackage{amsthm}
\usepackage{amsfonts}

\newtheorem{theorem}{Theorem}
\newtheorem{lemma}[theorem]{Lemma}
\newtheorem{corollary}[theorem]{Corollary}
\newtheorem{definition}[theorem]{Definition}

\begin{document}

\title{An Approach to Circuit Lower Bounds via Bounded Width Circuits}

\author{Hiroki Morizumi\\
{\small Shimane University, Japan}\\
{\small morizumi@cis.shimane-u.ac.jp}
}

\date{}

\maketitle

\begin{abstract}
In this paper, we investigate an approach to circuit lower bounds
via bounded width circuits.
The approach consists of two steps:
(i) We convert circuits to (deterministic or nondeterministic)
bounded width circuits.
(ii) We prove lower bounds for the bounded width circuits.

For the second step, we prove that there is an explicit
Boolean function $f$ as follows:
If a nondeterministic circuit of size $s$ and width $w$ computes $f$, then
$w = \Omega(\frac{n^4}{4^{\frac{s}{n}}s^3}) - \frac{\log_2 s}{2}$.
For the first step, we give some observations, which include
a relation between conversions to bounded width circuits and
a standard pebble game.
\end{abstract}

\section{Introduction}

Proving circuit lower bounds is one of the central topics in computational
complexity theory.
In this paper, we propose and investigate an approach to circuit lower bounds
via bounded width circuits.
The approach consists of two steps.

\begin{enumerate}
\item We convert circuits to (deterministic or nondeterministic)
      bounded width circuits.
\item We prove lower bounds for the bounded width circuits.
\end{enumerate}

\noindent In this paper, we give some results for each step.

\medskip

\noindent {\bf Note.}
The definition of the width in circuits is given in Section~\ref{subsec:circ}.
We consider bounded width circuits also for the nondeterministic case.
To the best of our knowledge, there has been no paper which considered
nondeterministic bounded width circuits, and note our definition
of nondeterministic bounded width circuits.
See Section~\ref{subsec:nbwc} for the details.

\subsection{Lower Bounds for Bounded Width Circuits} \label{subsec:intro_lb}

We firstly prove lower bounds for bounded width circuits, which are
related to the second step of the approach.
Our result for lower bounds is the following one theorem.

\begin{theorem} \label{thrm:lb}
There is an explicit Boolean function $f$ as follows: If a nondeterministic
circuit of size $s$ and width $w$ computes $f$, then
$$w = \Omega(\frac{n^4}{4^{\frac{s}{n}}s^3}) - \frac{\log_2 s}{2}.$$
\end{theorem}

\noindent By the theorem, we can obtain lower bounds for bounded width circuits
as follows, for example.

\begin{corollary} \label{coro:lb1}
There is an explicit Boolean function $f$ which cannot be computed by any
nondeterministic circuit of size $O(n)$ and width $o(n)$.
\end{corollary}
\begin{corollary} \label{coro:lb2}
There is an explicit Boolean function $f$ which cannot be computed by any
nondeterministic circuit of size $O(n \log\log n)$ and
width $\frac{n}{\log^{\omega(1)} n}$.
\end{corollary}

The proof of our lower bounds is based on a known result of lower bounds for
branching programs, which is in Section~\ref{subsec:lbbp}.

\smallskip

\noindent {\bf Bounded width circuits.}
Although bounded width circuits are an intermediate model in the main aim
of this paper (i.e., the approach to circuit lower bounds),
lower bounds for bounded width circuits may be of independent interest.
The width of circuits corresponds to the space of Turing machines, and
nonuniform-{\sf SC} is the class of decision problems solvable by
Boolean circuits with polynomial size, polylogarithmic width, and fan-in 2.
({\sf SC} is the class of decision problems solvable by a Turing machine
that simultaneously uses polynomial time and polylogarithmic space.)
Nevertheless, lower bounds for bounded width circuits have not been
discussed previously as far as we know.

\smallskip

\noindent {\bf Related problems.}
Proving that there is an explicit Boolean function $f$ which cannot be
computed by any (deterministic) circuit of size $O(n)$ and depth $O(\log n)$
is one of the realistic goals in circuit complexity.
Corollary~\ref{coro:lb1} resolves the width variant of the open problem
in a stronger form.
Proving a superlinear size lower bound for general circuits is a central
problem in circuit complexity.
Corollary~\ref{coro:lb1} implies that we can prove a superlinear
lower bound if the width is slightly bounded.

\subsection{Conversions to Bounded Width Circuits} \label{subsec:intro_conv}

We have proved lower bounds for bounded width circuits in the previous
subsection.
Theorem~\ref{thrm:lb} show the goal of conversions in the first step
of the approach.
(The goal is the current goal and improved lower bounds are hoped.)
By Theorem~\ref{thrm:lb}, we obtain the following corollary.

\begin{corollary} \label{coro:conv1}
If any circuit in a circuit class $\mathcal{C}$ can be converted to
a nondeterministic circuit which does not satisfy
$$w = \Omega(\frac{n^4}{4^{\frac{s}{n}}s^3}) - \frac{\log_2 s}{2},$$
then there is an explicit Boolean function $f$ which cannot be computed
by any circuit in $\mathcal{C}$.
\end{corollary}

\noindent The following corollary is more explicit.

\begin{corollary} \label{coro:conv2}
If any circuit in a circuit class $\mathcal{C}$ can be converted to
a nondeterministic circuit which satisfy
size $O(n \log\log n)$ and width $\frac{n}{\log^{\omega(1)} n}$,
then there is an explicit Boolean function $f$ which cannot be computed
by any circuit in $\mathcal{C}$.
\end{corollary}

In Section~\ref{sec:conv}, we give some observations, which include
a relation between conversions to bounded width circuits and
a standard pebble game.

\subsection{Paper Organization}

Section~\ref{sec:lb} and Section~\ref{sec:conv} correspond to
Section~\ref{subsec:intro_lb} and Section~\ref{subsec:intro_conv},
respectively.
Our proof outline for lower bounds also provides a satisfiability algorithm
for nondeterministic bounded width circuits.
We mention the algorithm in Section~\ref{sec:conc}.

\section{Preliminaries}

\subsection{Circuits} \label{subsec:circ}

{\em Circuits} are formally defined as directed acyclic graphs.
The nodes of in-degree 0 are called {\em inputs}, and each one of them
is labeled by a variable or by a constant 0 or 1.
The other nodes are called {\em gates}, and each one of them
is labeled by a Boolean function.
The {\em fan-in} of a node is the in-degree of the node, and
the {\em fan-out} of a node is the out-degree of the node.
In this paper, the gates are AND gates of fan-in two, OR gates of fan-in two,
and NOT gates.
There is a single specific node called {\em output}.
The {\em size} of a circuit is the number of gates in the circuit.

When we consider the {\em width} of a circuit, we temporarily
insert COPY gates to the circuit.
A COPY gate is a dummy gate which simply outputs its input.
A circuit is {\em layered} if its set of gates can be partitioned
into subsets called {\em layers} such that every edge in the circuit
is between adjacent layers.
Note that every circuit is naturally converted to a layered circuit
by inserting COPY gates to each edge which jumps over some layers.
The width of a layer is the number of gates in the layer.
The width of a circuit is the maximum width of all layers in
the circuit.

A {\em nondeterministic circuit} is a circuit with {\em actual inputs}
$(x_1, \ldots, x_n) \in \{0,1\}^n$ and some further inputs
$(y_1, \ldots, y_m) \in \{0,1\}^m$ called {\em guess inputs}.
A nondeterministic circuit computes a Boolean function $f$ as follows:
For $x \in \{0,1\}^n$, $f(x)=1$ iff there exists a setting of the guess inputs
$\{y_1, \ldots, y_m\}$ which makes the circuit output 1.
We call a circuit without guess inputs
a {\em deterministic circuit} to distinguish it
from a nondeterministic circuit.

\subsection{The Definition of Nondeterministic Bounded Width Circuits} \label{subsec:nbwc}

To the best of our knowledge, there has been no paper which consider
nondeterministic bounded width circuits, and we need to notice that
the appearance of guess inputs is a sensitive problem to the
computational power of the circuits.
We restrict the number of nodes labeled by a guess input to at most one,
and we do not restrict the number of nodes labeled by an actual input.
(The restriction makes no sense if the width is unbounded.)
Also in nondeterministic Turing machines, each of nondeterministic choices
cannot be repeated and Turing machines cannot store all nondeterministic
choices if the space of Turing machines is bounded.

\subsection{Branching Programs}

A {\em nondeterministic branching program} is a directed acyclic graph.
The nodes of non-zero out-degree are called {\em inner nodes} and labeled
by a variable.
The nodes of out-degree 0 are called {\em sinks} and labeled by 0 or 1.
For each inner node, outgoing edges are labeled by 0 or 1.
There is a single specific node called the {\em start node}.
The output of the nondeterministic branching program is 1 if and only if
at least one path leads to 1 sink.
The {\em size} of branching programs is the number of its nodes.
A branching program is {\em syntactic read-$k$-times} if each variable
appears at most $k$ times in each path.

\section{Lower Bounds for Bounded Width Circuits} \label{sec:lb}

\subsection{Lower Bounds for Branching Programs} \label{subsec:lbbp}

To prove Theorem~\ref{thrm:lb}, we use the following theorem.

\begin{theorem}[\cite{BRS93}] \label{thrm:brs}
There is an explicit Boolean function $f$ such that every nondeterministic
syntactic read-$k$-times branching program for computing $f$ has
size $\exp(\Omega(\frac{n}{4^kk^3}))$.
\end{theorem}

\noindent In this paper, we denote by $f_{BRS}$ the Boolean function of
the theorem.

\subsection{Proof of Theorem~\ref{thrm:lb}} \label{subsec:lbpr}

In this subsection, we prove Theorem~\ref{thrm:lb}.
Firstly, we define the concept of {\em read-$k$-times} for circuits.
A variable is read-$k$-times in a circuit if the number of nodes labeled
by the variable is at most $k$.
A circuit is read-$k$-times if every actual input is read-$k$-times
in the circuit.
(Note that ``read-$k$-times'' makes a sense since the width of circuits
is bounded.)

\begin{lemma} \label{lem:conv_rkt}
Any nondeterministic read-$k$-times circuit of size $s$ and width $w$
can be converted to a nondeterministic syntactic read-$k$-times
branching program of size $4^ws$.
\end{lemma}

\begin{proof}
The number of values (0 or 1) from a layer to the next layer is
at most $w$.
For each at most $2^w$ combination of 0 and 1, we prepare one node
in the constructed branching program.
Natural conversion from the circuit to the branching program is
enough to prove the lemma.
\end{proof}

If the size of a circuit is $s$, then the number of nodes labeled by
actual inputs is at most $s+1$.
Thus, the average number of nodes labeled by each actual input is
at most $\frac{s+1}{n}$.
However, the circuit is not necessarily read-$\frac{s+1}{n}$-times.
This is the most difficult point of this proof.
We resolve the difficulty by the definition of a Boolean function $f$.

We define $f(x_1, x_2, \ldots, x_{2n}, z_1, z_2, \ldots, z_{2n})$ as follows. 
If $\sum_{i=1}^{2n} z_i \neq n$, then $f = 0$.
Otherwise, $f = f_{BRS}$ and the $n$ input variables are $x_i$'s such that
$z_i = 1$.

\begin{proof}[Proof of Theorem~\ref{thrm:lb}]
Let $C$ be a nondeterministic circuit computing
$f(x_1, x_2, \ldots, x_{2n}, z_1, z_2, \ldots, z_{2n})$, and let $s$ and $w$
be the size and the width of $C$, respectively.
We choose $n$ variables from $x_1, \ldots, x_{2n}$ so as every chosen
variable is read-$\frac{s}{n}$-times in $C$.
We assign $1$ to $z_i$ if $x_i$ has been chosen for $1 \leq i \leq 2n$,
and assign $0$ otherwise.
We assign an arbitrary value to $x_i$ which has not been chosen
for $1 \leq i \leq 2n$.
Let $C'$ be the obtained read-$\frac{s}{n}$-times circuit, and let $s'$ and $w'$
be the size and the width of $C'$, respectively.
$C'$ computes $f_{BRS}$, and $s' \leq s$ and $w' \leq w$.
By Lemma~\ref{lem:conv_rkt} and Theorem~\ref{thrm:brs},
\begin{eqnarray*}
  4^{w'}s'       & = & \exp(\Omega(\frac{n}{4^{\frac{s}{n}}(\frac{s}{n})^3})) \\
  4^ws          & = & \exp(\Omega(\frac{n}{4^{\frac{s}{n}}(\frac{s}{n})^3})) \\
  2w + \log_2 s & = & \Omega(\frac{n^4}{4^{\frac{s}{n}}s^3}) \\
  w             & = & \Omega(\frac{n^4}{4^{\frac{s}{n}}s^3}) - \frac{\log_2 s}{2}.
\end{eqnarray*}
\end{proof}

\section{Conversions to Bounded Width Circuits} \label{sec:conv}

We can choose various circuits as a circuit class $\mathcal{C}$
in Corollary~\ref{coro:conv1} and Corollary~\ref{coro:conv2}.
In this section, we firstly show an example of conversions from
restricted circuits to bounded width circuits, and prove
a relation between conversions to bounded width circuits and
a standard pebble game.

\subsection{Conversions from $O(n)$-size $O(\log n)$-depth circuits}

Proving that there is an explicit Boolean function $f$ which cannot be
computed by any deterministic circuit of size $O(n)$ and depth $O(\log n)$
is one of the realistic goals in circuit complexity.
The following theorem is an example of conversions from restricted circuits
to bounded width circuits.

\begin{theorem} \label{thrm:conv}
Any deterministic circuit of size $O(n)$ and depth $O(\log n)$ can be
converted to a deterministic circuit of size $O(n^{1 + \epsilon})$ and
width $O(n / \log\log n)$.
\end{theorem}
\begin{proof}
Let $C$ be a circuit of size $O(n)$ and depth $O(\log n)$.
It is known that we can find $O(n / \log\log n)$ edges in $C$ whose
removal yields a circuit of depth at most $\epsilon \log n$
(\cite{V76}, Section~14.4.3 of \cite{AB09}).
We construct a circuit which computes the value of $O(n / \log\log n)$ edges
one by one.
\end{proof}

Although we can obtain a bounded width circuit by Theorem~\ref{thrm:conv},
the bounded width circuit is large for our current goal which is
mentioned in Section~\ref{subsec:intro_conv}.

\subsection{A Pebble Game and Conversions}

The pebble game considered in this paper is the following one.
(For the definition, we referred to the paper~\cite{N13}.)

\begin{definition}[Pebble game]
Let $G$ be a directed acyclic graph with a unique sink vertex $z$.
The {\em black-white pebble game} on $G$ is the following one-player game.
At any time $t$, we have a configuration $\mathbb{P}_t = (B_t, W_t)$ of
black pebbles $B_t$ and white pebbles $W_t$ on the vertices of $G$,
at most one pebble per vertex.
The rules of the game are as follows:
\begin{enumerate}
\item If all immediate predecessors of an empty vertex $v$ have pebbles
on them, a black pebble may be placed on $v$. In particular, a black pebble
can always be placed on a source vertex.
\item A black pebble may be removed from any vertex at any time.
\item A white pebble may be placed on any empty vertex at any time.
\item If all immediate predecessors of a white-pebbled vertex $v$ have pebbles
on them, the white pebble on $v$ may be removed. In particular, a white pebble
can always be removed from a source vertex.
\end{enumerate}
A {\em black-white pebbling} of $G$, is a sequence of pebble configurations
$\mathcal{P} = \{\mathbb{P}_0, \ldots, \mathbb{P}_{\tau}\}$ such that
$\mathbb{P}_0 = (\emptyset, \emptyset)$,
$\mathbb{P}_{\tau} = (\{z\}, \emptyset)$, and for all $t \in [\tau]$,
$\mathbb{P}_t$ follows from $\mathbb{P}_{t-1}$ by one of the rules above.
The {\em time} of a pebbling
$\mathcal{P} = \{\mathbb{P}_0, \ldots, \mathbb{P}_{\tau}\}$ is simply
$\mathit{time}(\mathcal{P}) = \tau$ and the {\em space} is
$\mathit{space}(\mathcal{P}) = \max_{0 \leq t \leq \tau} \{|B_t \cup W_t|\}$.
A {\em black pebbling} is a pebbling using black pebbles only, i.e.,
having $W_t = \emptyset$ for all $t$.
\end{definition}

Since circuits have been defined as directed acyclic graphs, we can consider
the pebble game for circuits.
(The inputs and the output in circuits are source vertices and a unique sink
in graphs.)
The following two theorems show a relation between the pebble game and
conversions to bounded width circuits.

\begin{theorem} \label{thrm:peb1}
If there is a black pebbling $\mathcal{P}$ for
a deterministic circuit $C$, then the circuit $C$ can be converted to
a deterministic circuit of size $\mathit{time}(\mathcal{P})$ and
width $\mathit{space}(\mathcal{P})$.
\end{theorem}

\begin{proof}
For the placement and removal of black pebbles,
the circuit is constructed as follows:
\begin{enumerate}
\item If a black pebble is placed, then an input node or gate node
is added and connected from other nodes appropriately.
\item If a black pebble is removed, then we do nothing.
\end{enumerate}
The obtained circuit and $C$ obviously compute the same Boolean function.
The size of the circuit is at most the number of the placement
of black pebbles, which is at most $\mathit{time}(\mathcal{P})$.
For a time $t$, the added inputs and gates until $t$ are used from $t+1$
only if the placed black pebble is not removed until $t$.
Thus, the width of the layer of the circuit for the time $t$ is at most
the number of current black pebbles (i.e., $|B_t|$), and the width of
the circuit is at most $\mathit{space}(\mathcal{P})$.
\end{proof}

\begin{theorem}
If there is a black-white pebbling $\mathcal{P}$ for
a deterministic circuit $C$, then the circuit $C$ can be converted to
a nondeterministic circuit of size $c \cdot \mathit{time}(\mathcal{P})$ and
width $\mathit{space}(\mathcal{P}) + 1$ for a constant $c$.
\end{theorem}

\begin{proof}
For each of the placement and removal of black pebbles and white pebbles,
the circuit is constructed as follows:
\begin{enumerate}
\item If a black pebble is placed, then an actual input node or gate node
is added and connected from other nodes appropriately.
\item If a black pebble is removed, then we do nothing.
\item If a white pebble is placed, then a guess input node (labeled by
a new guess variable) is added.
\item If a white pebble is removed, then some gate nodes (and exceptionally
input nodes) are added to check the correctness of the guess input.
(The details are below.)
\end{enumerate}
The obtained circuit satisfies the conditions of the theorem, which is
proved by the similar way of Theorem~\ref{thrm:peb1}.

In the obtained circuit, each guess input guesses the output of
the corresponding gate (or exceptionally input).
In the case 4, the correctness of guess inputs is checked.
If the guess is not correct, then the output of the circuit is forcibly 0,
which guarantees the obtained circuit to compute the correct Boolean function
by the definition of nondeterministic circuits.
In the case 4, the number of added gates is at most constant, and
the width of the circuit increases by 1 to make the output 0 forcibly.
\end{proof}

\section{Concluding Remarks} \label{sec:conc}

In this paper, we investigated an approach to circuit lower bounds
via bounded width circuits.
The current state of the approach does not arrive at new circuit
lower bounds.
Although we have not found a sufficient pebbling even for restricted graphs,
known results of the pebble game may be helpful.

The proof outline for lower bounds in Section~\ref{subsec:lbpr} also
provides a satisfiability algorithm for nondeterministic bounded width circuits.
Finally, we mention the algorithm.

\subsection{Satisfiability Algorithms for Bounded Width Circuits}

The circuit satisfiability problem is defined as follows.

\medskip

\noindent {\bf Input:} A Boolean circuit $C$

\smallskip

\noindent {\bf Output:} Does $C$ output $1$ on some input?

\medskip

\noindent The circuit satisfiability problem has been intensively studied since
Ryan Williams showed a connection between the problem and lower bounds for
circuit complexity~\cite{Wil14}.

A satisfiability algorithm for nondeterministic syntactic
read-$k$-times branching programs has been provided as follows.

\begin{theorem}[\cite{NST17}] \label{thrm:nst}
There exists a deterministic and polynomial space algorithm for
a nondeterministic and syntactic read-$k$-times BP SAT with $n$ variables
and $m$ edges that runs in time
$O(\mathrm{poly}(n, m^{k^2}) \cdot 2^{(1 - 4^{-k-1})n})$.
\end{theorem}

\noindent In a similar outline to the proof of the lower bound,
a satisfiability algorithm for nondeterministic bounded width circuits is
provided.

\begin{theorem}
There exists a deterministic and polynomial space algorithm for
a nondeterministic bounded width circuit SAT with $n$ actual inputs,
size $s$ and width $w$ that runs in time
$O(\mathrm{poly}(n, m^{k^2}) \cdot 2^{(1 - 4^{-k-\frac{3}{2}})n})$,
where $m = 4^ws$ and $k = \lceil \frac{2s}{n} \rceil$.
\end{theorem}

\begin{proof}
Let $C$ be a nondeterministic bounded width circuit.
We choose $n/2$ variables from $x_1, \ldots, x_n$ so as every chosen
variable is read-$\lceil \frac{2s}{n} \rceil$-times in $C$.
For $n/2$ variables which have not been chosen, we execute
the brute-force search.
Then, we use Lemma~\ref{lem:conv_rkt}, and execute the algorithm
in Theorem~\ref{thrm:nst}.
\end{proof}

\bibliographystyle{plain}
\bibliography{circuit}

\begin{thebibliography}{1}

\bibitem{AB09}
Sanjeev Arora and Boaz Barak.
\newblock {\em Computational Complexity - {A} Modern Approach}.
\newblock Cambridge University Press, 2009.

\bibitem{BRS93}
Allan Borodin, Alexander~A. Razborov, and Roman Smolensky.
\newblock On lower bounds for read-k-times branching programs.
\newblock {\em Computational Complexity}, 3:1--18, 1993.

\bibitem{NST17}
Atsuki Nagao, Kazuhisa Seto, and Junichi Teruyama.
\newblock Satisfiability algorithm for syntactic read-$k$-times branching
  programs.
\newblock In {\em Proc. of ISAAC}, pages 58:1--58:10, 2017.

\bibitem{N13}
Jakob Nordstr{\"{o}}m.
\newblock Pebble games, proof complexity, and time-space trade-offs.
\newblock {\em Log. Methods Comput. Sci.}, 9(3), 2013.

\bibitem{V76}
Leslie~G. Valiant.
\newblock Graph-theoretic properties in computational complexity.
\newblock {\em J. Comput. Syst. Sci.}, 13(3):278--285, 1976.

\bibitem{Wil14}
Ryan Williams.
\newblock Algorithms for circuits and circuits for algorithms.
\newblock In {\em Proc. of CCC}, pages 248--261, 2014.

\end{thebibliography}

\end{document}